\documentclass[a4paper,german,12pt]{article}
\usepackage[T1]{fontenc}
\usepackage[latin1]{inputenc}
\usepackage[english]{babel}
\usepackage{graphicx}
\usepackage{amssymb}
\usepackage{amsthm}
\usepackage{amsmath}
\usepackage{flafter}
\usepackage{placeins}
\usepackage{pst-all}
\usepackage{multido}
\usepackage{mathrsfs}
\makeatletter

\begin{document}
\newcommand{\f}[2]{#1 &$\pm$ #2}
\newcommand{\feh}[1]{\left(\frac{\Delta #1}{#1}\right)^2}
\newcommand{\fehleri}[2]{\Delta #1 = #1 \frac{\Delta #2}{#2}}
\newcommand{\fehlerii}[3]{\Delta #1 = #1 \sqrt{\feh{#2}+\feh{#3}}}
\newcommand{\fehleriv}[5]{\Delta #1 = #1 \sqrt{\feh{#2}+\feh{#3}+\feh{#4}+\feh{#5}}}
\newcommand{\folge}{ \quad \longrightarrow \quad }
\newcommand{\Folge}{ \quad \Rightarrow \quad }
\newcommand{\degr}{^{\circ}}
\newcommand{\cel}{\,\degr C}
\newcommand{\pabl}[2]{\frac{\partial #1}{\partial #2}}
\newcommand{\pms}{\,\pm\,}
\newcommand{\pmss}{\;\pm\;}
\newcommand{\ket}[1]{\left| #1 \right>}
\newcommand{\bra}[1]{\left< #1 \right|}
\newcommand{\braket}[2]{\left< #1 | #2 \right>}

\newcommand{\bild}[3]{\begin{figure}[h]
\center
\includegraphics[angle=0,width=#2cm]{../bilder/#1.eps}
\caption{#3}
\label{#1}
\end{figure}}

\newcommand{\diag}[3]{\begin{figure}[h]
\center
\includegraphics[angle=-90,width=#2cm]{../daten/#1.eps}
\caption{#3}
\label{#1}
\end{figure}}

\newcommand{\beqa}{\begin{eqnarray*}}
\newcommand{\eeqa}{\end{eqnarray*}}
\newcommand{\beqan}{\begin{eqnarray}}
\newcommand{\eeqan}{\end{eqnarray}}
\newcommand{\beq}{\begin{equation}}
\newcommand{\eeq}{\end{equation}}
\newcommand{\id}{1\!\! 1}
\newcommand{\is}{\!\!\!&=&\!\!\!}
\newcommand{\bsp}{\!\!\!\!\!\!}
\newcommand{\eps}{\varepsilon}
\newcommand{\rar}{\rightarrow}
\newcommand{\Rar}{\Rightarrow}
\newcommand{\lrar}{\longrightarrow}
\newcommand{\uli}{\underline}
\newcommand{\oli}{\overline}
\newcommand{\equi}{ \quad \Leftrightarrow \quad }
\newcommand{\tdot}[1]{\dot{\!\ddot{#1}}}
\newtheorem{Def}{Definition}[section]
\newtheorem{Satz}{Theorem}[section]
\newtheorem{Lemma}{Lemma}[section]
\newtheorem{Proposition}{Proposition}[section]

\title{Dynamical Backreaction in Robertson-Walker Spacetime}
\author{
Benjamin Eltzner, Hanno Gottschalk}
\date{February 2011}

\maketitle

\begin{abstract}
The treatment of a quantized field in a curved spacetime requires the study of backreaction of the field on the spacetime via the semiclassical Einstein equation. We consider a free scalar field in spatially flat Robertson-Walker space time. We require the state of the field to allow for a renormalized semiclassical stress tensor. We calculate the sigularities of the stress tensor restricted to equal times in agreement with the usual renormalization prescription for Hadamard states to perform an explicit renormalization. The dynamical system for the Robertson Walker scale parameter $a(t)$ coupled to the scalar field is finally derived for the case of conformal and also general coupling.  
\end{abstract}

\section{Introduction}
The studies of quantized fields in curved spacetimes usually assume a fixed background spacetime on which a quantized field is defined \cite{Ful}. Such a setting has also been used to investigate the simplest cosmological model of a homogeneous isotropic universe. In this context the mode spectrum of a scalar field in a homogeneous isotropic universe undergoing an era of inflation has been found to be the famous and experimentally confirmed scale free Harrison-Zeldovich spectrum \cite{HZ}. 

The work which is concerned with the backreaction of the field on the evolution of the spacetime itself usually takes a mean field approach in which quantum fluctuations only contribute to the effective potential of the (classical) expectation value of the field \cite{Lin,Weinb}. The dynamics of the quantum degrees of freedom is not considered. In other approaches, the quantum degrees of freedom either decouple from the spacetime \cite{Star} due to conformal coupling for a massless field, or a large mass of the field is assumed \cite{DFP}, which implies another semiclassical approximation. This approximation in fact implies that the field configuration is dominated by renormalization ambiguities. Again the coupling of quantum degrees of freedom to the geometry of the space time is only approximate.  

In this work, which is partially based on \cite{Elz}, we determine the coupling of a free quantum field to the scale parameter of the spatially flat Robertson-Walker spacetime. The result is an equation for a dynamical system with infinitely many degrees of freedom that can, at least in principle, be solved. 

The dynamical system is derived via an expansion of Riemann normal coordinates in the Robertson-Walker canonical coordinates up to third order \cite{Gack}. This allows to calculate the singular terms of the Hadamard bidistribution and its first and second time derivative restricted to equal time surfaces up to constant terms (zero mode terms) of the second time derivative. Those terms are needed to renormalize the energy-momentum tensor of the free field restricted to equal time surfaces. The latter establishes the dynamics of the Robertson-Walker scale factor (up to renormalization ambiguities).

While finalizing this paper we got aware of the publication \cite{Pina}, which derives exactly the same equation of motion as we do for the case of conformal coupling using a completely independent argument. In fact, this work in the conformally coupled case also establishes the existence of solutions for small time intervals given that the initial state fulfills certain conditions, which goes beyond the scope of this article.  

\section{General description of dynamics and renormalization approach}
Throughout the article we restrict to flat Robertson-Walker (RW) spacetimes with space-time dimension 4. This allows to use the standard Fourier transform on spatial sections of constant RW time in order to formulate the dynamics of the field in terms of modes. As the stress tensor is formulated as a differential operator acting on the two-point function, we will formulate the dynamics of the field in such a way that the two-point function of the field and its associate momentum field operator restricted to the given time is the dynamical variable.  We will also restrict our attention to homogeneous, isotropic, quasifree and pure states that either are Hadamard states or are sufficiently close to Hadamard in the sense that they allow the same renormalization prescription for their energy-momentum tensor. Hadamard states and adiabatic vacuum states on RW space time have been studied e.g. by Lüders and Roberts \cite{LR}, Juncker and Schrohe \cite{JS} and more recently by Olbermann \cite{Olb}.  

The line element of a spatially flat Robertson-Walker metric for a homogeneous, isotropic spacetime is
\beq
ds^2 = dt^2 - a^2(t) d \vec{x}^2
\eeq
where we call $a(t)$ the scale parameter and define the Hubble parameter $H(t) = \frac{\dot{a}}{a}$. We consider a scalar free field of mass $m$ coupled to the scalar curvature $R(t)= - 6 \left(\dot{H}(t) + 2 H^2(t) \right)$ with coupling $\xi$,
\beq
(\Box-\xi R+m^2)\phi=0,
\eeq
where $\Box$ is the D'Alembertian.

We use the abbreviation $\omega_k^2 = a^{-2} k^2  + m^2 - \xi R$. The Klein-Gordon equation for the field modes $\phi_k$ is given by
\beq
\ddot{\phi}_k + 3 H \dot{\phi}_k + \omega_k^2 \phi_k = 0
\eeq
Using the canonical momentum field $\pi_k = a^3 \dot{\phi}_k$ we consider the Hamiltonian form of this equation
\beq
\partial_t \left( \begin{array}{c} \phi_k\\ \pi_k \end{array} \right) = \left( \begin{array}{cc} 0	&	a^{-3}\\	-a^{3} \omega_k^2	&	0 \end{array} \right) \left( \begin{array}{c} \phi_k\\ \pi_k \end{array} \right)
\label{mKGG1}\eeq

It is shown in \cite{LR} that the equal time two-point function, i.e. the two-point function on a Cauchy surface, of a state can be described by matrices 
\beq
\left( \begin{array}{cc} G_{\phi\phi,k} & G_{\phi\pi,k}\\ G_{\pi\phi,k} & G_{\pi\pi,k} \end{array} \right),
\label{Statedef}
\eeq
where $G_{\phi\phi,k}=G_{\phi\phi}(t,k)$ is the Fourier transform in $\vec z=\vec x-\vec y$ of two equal time field operators $G(t,\vec x,\vec y)=\omega(\phi(t,\vec x)\phi(t,\vec y))$, $G_{\phi\pi,k}$ the Fourier transform of one equal time field operator and one canonically conjugated momentum operator etc., and $k$ is the modulus of the momentum conjugated to $\vec z$. Our normalization convention for the Fourier transform is $\mathscr{F}(f)(\vec k)=\int_{\mathbb{R}^3} f(\vec x) e^{i\vec k\cdot \vec x}dx$. The positivity of the state will enforce that the matrix is positive semidefinite with its determinant vanishing for a pure quasifree state.  

The symmetric part of the two-point function modes then fulfills the linear system of equations
\beq
\label{modeMotion}
\partial_t \left( \begin{array}{c} G_{\phi\phi,k}\\ G_{(\phi\pi),k}\\ G_{\pi\pi,k} \end{array} \right) = \left( \begin{array}{ccc} 0	&	2 a^{-3}	&	0\\	-a^{3} \omega_k^2	&	0 & a^{-3}\\	0	&	-2 a^{3} \omega_k^2	&	0 \end{array} \right) \left( \begin{array}{c} G_{\phi\phi,k}\\ G_{(\phi\pi),k}\\ G_{\pi\pi,k} \end{array} \right)\eeq

Here and in the following $()$ stands for symmetrization and $[\,]$ for anti symmetrization in the field and momentum operator. This system of equations has one conserved quantity per mode
\beq
J_k = G_{\phi\phi,k} G_{\pi\pi,k} - G^2_{(\phi\pi),k}
\eeq
which reduces the number of degrees of freedom per mode to two, as required.

The condition for a state to be pure and to induce a representation of the canonical commutation (CCR) algebra then is 
\beq
\forall k: \; J_k = - G^2_{[\phi\pi],k} = \frac{1}{4}
\label{sCCR}\eeq
which implies the vanishing of the above-mentioned determinant for all modes.

Next we turn to the semi classical Einstein equation. This equation has the form
\beq
\label{semiclassicEEQ}
G_{\mu\nu} = 8 \pi G \left< T_{\mu\nu} \right>_{\omega}
\eeq
where the expectation value of the stress tensor must be renormalized. In this equation $G_{\mu\nu}$ is the Einstein tensor and $G$ is the gravitational constant and both should not be confused with $G_{\phi\phi,k}$ etc. which are two-point functions.

We will apply the point-splitting procedure as formulated in \cite{Mor}, see also \cite{WalB}.  In order to preserve general covariance of the renormalization prescription, we need to subtract the Hadamard bidistribution $\mathfrak{H}$, described in some detail in \cite{Ful}, from the two-point function before removing the point splitting. Actually, a sufficiently precise approximation $\mathfrak{H}_n$ of the Hadamard parametrix does the job as well
\beq
\left< T_{\mu\nu}^{(\eta)} (v) \right>_{\omega,\lambda,\xi} = \lim_{(x,y)\to (v,v)}D^{(\eta)}_{(v)\mu\nu}(x,y) \left[\left< \phi(x) \phi(y) \right>_{\omega}-\mathfrak{H}_n(x,y)\right]+t_{\mu\nu}(v).
\label{EMTensor}
\eeq
with , cf \cite{Ha,WalB},
\beq
\label{reno}
t_{\mu\nu} (v)= \frac{\delta}{\delta g^{\mu\nu}(v)} \int  \left(Am^4+Bm^2R+C R^2 + D R_{\alpha\beta} R^{\alpha\beta} \right)d_gx 
\eeq
where $A,B,C$ and $D$ are real valued renormalization degrees of freedom and $d_gx$ stands for the canonical volume form $\sqrt{|g|}dx$.

$D^{(\eta)}_{(v)\mu\nu}(x,y)$ is the symmetrization in $x$ and $y$ of the following second order partial differential operator (cf. \cite{Mor} Eq. (10) for the details, where we corrected some minor misprints) 
\beqan &\partial'_{x,\mu}\partial'_{y,\nu}-\frac{1}{2}g_{\mu\nu}(g^{\gamma,\delta}\partial'_{x,\gamma}\partial'_{y,\delta}-m^2) +\xi\left[(R_{\mu\nu}-\frac{1}{2} g_{\mu\nu}R) +2g_{\mu\nu}\left(\Box_x +g^{\gamma\delta}\partial'_{x,\gamma} \partial'_{y,\delta}\right) \right.&
 \nonumber\\
 &\left. -2\left(D_{x,\mu}'\partial_{x,\nu}' +\partial'_{x,\mu}\partial'_{y,\nu}\right)\right]-\eta g_{\mu\nu}(\Box_x-\xi R+m^2) &
\eeqan
Here $\partial'_{x,\mu}=\delta_\mu^\nu(v,x)\partial_{x,\nu}$, $\delta(v,x)$ being the geodesic transport and $D'_{x,\mu}$ analogously denotes a covariant derivative. The quantities $g,R$ and the Ricci tensor $R_{\mu\nu}$  are evaluated at $v$. Note here we use sign convention $(+---)$ instead of $(-+++)$ in \cite{Mor} which flips the sign in the last expression in the brackets (and some more in the text below).

The definition of $\mathfrak{H}_n$ for $d=4$ is 
\beq
\label{parametrix}
-\frac{1}{4\pi^2}\frac{u(x,y)}{\sigma(x,y)}-\frac{1}{4\pi^2}\log\left(-\frac{\sigma(x,y)}{\lambda^2}\right)\sum_{k=0}^n\frac{1}{k!} v_{k}(x,y)\sigma(x,y)^k.
\eeq
Here $\sigma(x,y)$ is the squared geodesic distance, $u(x,y)$ is the square root of the Van Vleck-Morette determinant (which for dimension 4 is $U_0(x,y)$ as defined in \cite{Mor} up to normalization). $u(x,y)$ and the $v_k(x,y)$ (corresponding to $U_{k+1}(x,y)$ of \cite{Mor}) are functions that can be determined by a recursive system of differential equations that depends exclusively on invariant quantities and the operator $\Box-\xi R+m^2$:
\begin{align}
2 g^{\mu\nu}(x) (\partial^x_{\mu} \sigma) (\partial^x_{\nu} u) + (\square^x \sigma - 4) u &= 0\\
2 g^{\mu\nu}(x) (\partial^x_{\mu} \sigma) (\partial^x_{\nu} v_0) + (\square^x \sigma - 2) v_0 &= -(\square^x + m^2 - \xi R(x))u\\
2 g^{\mu\nu}(x) (\partial^x_{\mu} \sigma) (\partial^x_{\nu} v_{k+1}) + (\square^x \sigma + 2k) v_{k+1} &= -(\square^x + m^2 - \xi R(x))v_k
\end{align}
The coincidence limits $v_k(v,v)$ up to normalization coincide with the Hadamard-Minkashisundram-De Witt-Seely coefficients \cite{Ful,Mor}. In particular, on RW space time, $v_k((t,\vec x),(t,\vec y))$ depends on $\vec x$ and $\vec y$ only through $\vec z^2=(\vec x-\vec y)^2$ and the coincidence limit $v_k((t,\vec x),(t,\vec x))$ of the $v_k$ depends only on $t$. The same hold true for derivatives of $v_k$ wrt $\vec z^2$.

The regularization at light like separated $x$, $y$ is done by adding an $i\epsilon(x^0-y^0)$ to $\sigma(x,y)$, letting $\epsilon \searrow 0$ and taking the real part. We do not need this in the following, as we will approach the coincidence limit from spatial directions exclusively.

It is shown that for $n\in\mathbb{N}$, the expression $\left[\left< \phi(x) \phi(y) \right>_{\omega}-\mathfrak{H}_n(x,y)\right]$ can be extended to a function in $C^n(C_v\times C_v)$ where $C_v$ is some convex normal neighborhood of $v$. The quantum averaged field fluctuations are defined as
\beq
\left<\phi^2\right>_{\omega,\lambda,\xi}=\lim_{(x,y)\to(v,v)}\left[\left< \phi(x) \phi(y) \right>_{\omega}-\mathfrak{H}_n(x,y)\right].
\eeq
where $n=0$ suffices if no derivatives of this quantity are required.
For $\eta\not=0$, the corresponding term in (\ref{EMTensor}) in the energy momentum tensor vanishes classically, but not quantum mechanically. For $\eta=1/3$, the so-defined energy-momentum tensor is conserved. For RW spacetimes and states of the form (\ref{Statedef}), the quantum averaged energy momentum tensor depends only on $t$ and is diagonal so that only energy conservation is non trivial  
\beq
\nabla^\mu\left<T_{\mu0}^{(1/3)}\right>_{\omega,\lambda,\xi} =\dot \rho+H(\rho+p)=0,
\eeq
where $\rho= \left<T_{00}^{(1/3,\lambda)}\right>_{\omega,\lambda,\xi}$ and $p=a^2\left<T_{jj}^{(1/3,\lambda)}\right>_{\omega,\lambda,\xi}$ are energy density and pressure, respectively, see e.g. \cite{Weinb}.  

Moretti also showed
\beqan
\label{traceeq}
g^{\mu\nu}\left<T_{\mu\nu}^{(1/3)}\right>_{\omega,\lambda,\xi}&=&\rho-3p=\left[-3\left(\frac{1}{6}-\xi\right)\Box+m^2\right]\left<\phi^2\right>_{\omega,\lambda,\xi}\nonumber\\
&+&\frac{1}{4\pi^2}v_1+cm^2+c'\, m^2R+c''\,\Box R
\eeqan
where $v_1=v_1(t,\xi,m^2)=v_1((t,\vec v),(t,\vec v))$ and $c,c',c''$ can be calculated from $A,B,C$ and $D$. The D'Alembert operator here does not require point splitting and hence on flat RW space time can be replaced by $d^2/dt^2+3Hd/dt$.  Interestingly, this equation can be seen as the equation of state and the right hand side gives the deviation of {\em quantum} matter from the state equation of hot (relativistic) matter $p=\rho/3$. For Minkowski space and $\omega$ the Minkowski vacuum, we expect $p=\rho=0$ which can be achieved through $\lambda^2=4e^{\frac{7}{4}-2\gamma}/m^2$ with $\gamma$ the Euler constant. $v_1$ has been calculated in \cite{DFP} for the case of flat RW space time as 
\beqan
\label{U2}
v_1&=&\frac{1}{60}\left(\dot H H^2+H^4\right)+\frac{1}{24}\left(\frac{1}{5}-\xi\right)\Box R
\\
&-&\frac{9}{2}\left(\frac{1}{6}-\xi\right)^2\left( \dot H^2+4H^2\dot H+4H^4\right)-\frac{m^4}{8}+\frac{1}{4}\left( \frac{1}{6}-\xi\right)m^2R.\nonumber
\eeqan 
In the case of a massless field with conformal coupling $m^2=0$, $\xi=1/6$, the last term on the right hand side of (\ref{traceeq}) is called the conformal anomaly. 

Combining (\ref{semiclassicEEQ}) and (\ref{traceeq}) one arrives at the following equation of motion
\beqan
\label{motion}
-R&=&8\pi G\left\{\left[-3\left(\frac{1}{6}-\xi\right)\left(\frac{d^2}{dt^2}+3H\frac{d}{dt}\right)+m^2\right]\left<\phi^2\right>_{\omega,\lambda,\xi}\right.\nonumber\\&&\left.~~~~~~~~~~~~~~~~~~~~~~~~~~~~~~~~~~~~~+\frac{1}{4\pi^2}v_1+cm^4+c'\,m^2R+c''\,\Box R\right\}
\eeqan

We would like to find an expression for $(\frac{d^2}{dt^2}+3H\frac{d}{dt})\left<\phi^2\right>_{\omega,\lambda,\xi}$ that does not include terms $\left<\ddot\phi\phi\right>_{\omega,\lambda,\xi}$  since only then second order time derivatives of the field $\phi$ do only occur on the left hand side of (\ref{mKGG1}). It is clear that $\frac{d}{dt}\left<\phi^2\right>_{\omega,\lambda,\xi}$  equals $\langle\phi\dot\phi\rangle_{\omega,\lambda,\xi}+\langle\dot \phi\phi\rangle_{\omega,\lambda,\xi}$. Let $h(t)=f(s,s')|_{s=s'=t}$. Then, $\ddot h(t)=(\partial^2_s+\partial^2_{s'}+2\partial_s\partial_{s'})f(s,s')|_{s=s'=t}$. Applying this to our problem and using the fact \cite[Lemma 2.1]{Mor} that 
\beq
\left<\phi(\Box-\xi R+m^2)\phi\right>_{\omega,\lambda,\xi}=\frac{3}{2\pi^2}\, v_1,
\eeq
we conclude that 
\beq
\left(\frac{d^2}{dt^2}+3H\frac{d}{dt}\right)\langle \phi^2\rangle_{\omega,\lambda,\xi}=2\langle\dot\phi^2\rangle_{\omega,\lambda,\xi}+2\left<\phi (a^{-2}\Delta+\xi R-m^2)\phi\right>_{\omega,\lambda,\xi}+\frac{3}{\pi^2} v_1.
\eeq  
Here $\Delta$ stands for the Laplacian on $\mathbb{R}^3$.  We summarize the discussion in the following theorem:

\begin{Satz}
\label{2.1thm}
The equation of motion for semi-classical Einstein equation on flat Robertson-Walker space time can be written as follows:
\beqan
\label{motion2}
-R&=&8\pi G\left\{\left(6\xi-1\right)\left(\langle\dot\phi^2\rangle_{\omega,\lambda,\xi}+a^{-2}\left<\phi \Delta\phi\right>_{\omega,\lambda,\xi}\right)\right.\\
&+&\left.\left[(2-6\xi)m^2-(1-6\xi)\xi R\right]\left<\phi^2\right>_{\omega,\lambda,\xi}+\frac{36\xi-5}{4\pi^2}v_1+cm^4+c'\, m^2R+c''\,\Box R\right\}\nonumber
\eeqan
with $R=-6(\dot H+2H^2)$, $\Box R=(\frac{d^2}{dt^2}+3H\frac{d}{dt})R$ and $v_1$ given by (\ref{U2}).
\end{Satz}

From the above analysis it is clear that we need to compute $\left<\phi^2\right>_{\omega,\lambda,\xi}$ and its second time derivative in terms of the functions $G_{\phi\phi,}$\ldots and $a$,$H$, $\dot H$, $R$, $\dot R$, $\ddot R$ \ldots  in order to combine (\ref{modeMotion}) and (\ref{motion2}) to a closed system of equations. 

Hence we need to calculate those terms in $\mathfrak{H}(x,y)|_{x^0=y^0=t}$, $\Delta_{\vec x}\mathfrak{H}(x,y)|_{x^0=y^0=t}$ and $\frac{\partial}{\partial x^0}\frac{\partial}{\partial y^0}\mathfrak{H}(x,y)|_{x^0=y^0=t}$ that either are singular or contribute a (time dependent) quantity to the spatial coincidence limit $\vec z=(\vec x-\vec y)\to 0$.

\section{Leading Terms of the Hadamard Distribution}

We start the task described at the end of the preceding section via a perturbative calculation of the singularities in the leading term $\mathfrak{H'}(x,y)=-(4 \pi^2)^{-1} \frac{u(x,y)}{\sigma(x,y)}$ of the Hadamard distribution.  In normal coordinates $X$ around $y$, $\sigma$ takes the form
\beq
2 \sigma = \eta_{\mu\nu} X^{\mu} X^{\nu}
\eeq
with $\eta_{\mu\nu}$ the Minkowski metric. Thus the computation is reduced to finding normal coordinates perturbatively in dependence of the canonical coordinates. We will need this expansion to fifth order to fix all singularities of $\lim_{y\rar x} \limits \mathfrak{H'}(x,y)$ and its first and second derivatives as well as all homogeneous terms. The normal coordinates have been calculated perturbatively to fifth order in \cite{Bre} (using computer algebra support), formula (11.12) to (11.16) for a general metric. We plug in
\beq
\partial^n_0\Gamma^{\alpha}_{\beta\gamma} = \delta^{\alpha}_0 \delta^{i}_{\beta} \delta^{j}_{\gamma} \delta_{ij} L_n a^2 + \delta^{\alpha}_i (\delta^{0}_{\beta} \delta^{i}_{\gamma}+\delta^{i}_{\beta} \delta^{0}_{\gamma}) H^{(n)}
\eeq
where $H^{(n)}$ denotes the nth derivative of $H$ with respect to time and the $L_n$ are
\beqan
L_0 \is H\\
L_1 \is \dot{H} + 2 H^2\\
L_2 \is \ddot{H} + 6 \dot{H}H + 4 H^3\\
L_3 \is \tdot{H} + 8 \ddot{H}H + 6 \dot{H}^2 + 24 \dot{H}H^2 + 8 H^4
\eeqan

Using these formulae we get
\beqan
2 \sigma \is z_0^2 - a^2 \vec{z}^2 \bigg\{ 1 + H z_0 + \frac{1}{3} (\dot{H} + H^2) z_0^2 + \frac{1}{12} H^2 a^2 \vec{z}^2 \nonumber \\
	&& + \frac{1}{12} (\ddot{H} + 2 \dot{H} H) z_0^3 + \frac{1}{12} (\dot{H}H + 2 H^3) a^2 \vec{z}^2 z_0  \nonumber\\
	&& + \frac{1}{180} (3 \tdot{H} + 6 \ddot{H}H + 2 \dot{H}^2 - 8 \dot{H}H^2 - 4 H^4) z_0^4 \nonumber \\
	&& + \frac{1}{360} (9 \ddot{H}H + 8 \dot{H}^2 + 74 \dot{H}H^2 + 48 H^4) a^2 \vec{z}^2 z_0^2 \nonumber \\
	&& + \frac{1}{360} (3 \dot{H}H^2 + 4 H^4) a^4 (\vec{z}^2)^2 \bigg\}+ \mathcal{O}(z^7) \label{worldFunctionExpansion}
\eeqan
where all time dependent terms are evaluated at $y_0$ and we use the abbreviation $z = x-y$.

\begin{Lemma}
$\sigma$ is symmetric under exchange of $x$ and $y$
\end{Lemma}

\begin{proof}
The proof is done by straightforward calculation, however on has to keep in mind that $a$, $H$ and its derivative all have a suppressed argument $y_0$ such that they have to be Taylor expanded. The terms of even orders in $z$ do not change signs so it has to be shown, that they are not affected by the terms from the Taylor expansion of the lower order coefficients. Using abbreviations $a_y=a(y_0)$ and the likes we get to the relevant orders
\begin{align}
\frac{a_x^2}{a_y^2} =& 1 + 2 H_y z_0 + (\dot{H}_y + 2H_y^2) z_0^2 + \frac{1}{3}(\ddot{H}_y + 6\dot{H}_y H_y + 4H_y^3) z_0^3 \nonumber\\
&+ \frac{1}{12}(\tdot{H}_y + 8\ddot{H}_y H_y + 6 \dot{H}_y^2 + 24 \dot{H}_y H_y^2 + 8H_y^4)z_0^4\\
-H_x \frac{a_x^2}{a_y^2} z_0 =& -H_y z_0 - (\dot{H}_y + 2H_y^2) z_0^2 - \frac{1}{2}(\ddot{H}_y + 6\dot{H}_y H_y + 4H_y^3) z_0^3 \nonumber\\
&- \frac{1}{6}(\tdot{H}_y + 8\ddot{H}_y H_y + 6 \dot{H}_y^2 + 24 \dot{H}_y H_y^2 + 8H_y^4)z_0^4\\
\frac{1}{3}(\dot{H}_x + H_x^2)\frac{a_x^2}{a_y^2}z_0^2 =& \frac{1}{3}(\dot{H}_y + H_y^2) z_0^2 + \frac{1}{3}(\ddot{H}_y + 4\dot{H}_y H_y + 2H_y^3) z_0^3 \nonumber\\
&+ \frac{1}{6}(\tdot{H}_y + 6\ddot{H}_y H_y + 4 \dot{H}_y^2 + 14 \dot{H}_y H_y^2 + 4H_y^4)z_0^4\\
-\frac{1}{12} (\ddot{H}_x + 2 \dot{H}_x H_x)\frac{a_x^2}{a_y^2} z_0^3 =& -\frac{1}{12}(\ddot{H}_y + 2\dot{H}_y H_y) z_0^3 \nonumber\\
&- \frac{1}{12}(\tdot{H}_y + 4\ddot{H}_y H_y + 2 \dot{H}_y^2 + 4 \dot{H}_y H_y^2)z_0^4\\
\frac{1}{12} H_x^2 \frac{a_x^4}{a_y^4} =& \frac{1}{12} H_y^2 + \frac{1}{6}(\dot{H}_y H_y + 2 H_y^3) z_0 \nonumber\\
&+ \frac{1}{12}(\ddot{H}_y H_y + \dot{H}_y^2 + 10 \dot{H}_y H_y^2 + 8 H_y^4)z_0^2\\
-\frac{1}{12} (\dot{H}_x H_x + 2 H_x^3)\frac{a_x^4}{a_y^4} z_0 =& -\frac{1}{12}(\dot{H}_y H_y + 2 H_y^3) z_0 \nonumber\\
&- \frac{1}{12}(\ddot{H}_y H_y + \dot{H}_y^2 + 10 \dot{H}_y H_y^2 + 8 H_y^4)z_0^2
\end{align}
Adding up the first four and the last two terms one gets
\begin{align}
&1 + H_y z_0 + \frac{1}{3} (\dot{H} + H^2) z_0^2 + \frac{1}{12} (\ddot{H} + 2 \dot{H} H) z_0^3\\
&\frac{1}{12} H_y^2 + \frac{1}{12}(\dot{H}_y H_y + 2 H_y^3) z_0
\end{align}
which proves the claim.
\end{proof}

To calculate $u$ we use the first equation of well know Hadamard recursion which is deduced from the Klein Gordon equation
\beq
2 g^{\mu\nu}(x) (\partial^x_{\mu} \sigma) (\partial^x_{\nu} u) + (\square^x \sigma - 4) u = 0
\eeq
where one has to take into account, that the derivatives are with respect to $x$ and the metric is evaluated at $x$ in the above formula, which leads to additional terms in the calculation. We make the ansatz
\beq
u = 1 + \mu z_0^2 + \nu \vec{z}^2 + \rho z_0^3 + \tau \vec{z}^2 z_0 + \phi z_0^4 + \psi \vec{z}^2 z_0^2 + \chi (\vec{z}^2)^2
\eeq
where all coefficients are evaluated at $y_0$. The first term being $1$ and the absence of a first order term follow from the requirements that $u(x,x)=1$ and $u(x,y) = u(y,x)$.

The calculation of the coefficients then yields 
\beqan
u \is 1 - \frac{1}{4} (\dot{H} + H^2) z_0^2 + \frac{1}{12} (\dot{H} + 3 H^2) a^2 \vec{z}^2 \nonumber \\
	&& - \frac{1}{8} (\ddot{H} + 2 \dot{H} H) z_0^3 + \frac{1}{24} (\ddot{H} + 8 \dot{H}H + 6 H^3) a^2 \vec{z}^2 z_0  \nonumber\\
	&& + \frac{1}{480} (- 18 \tdot{H} - 36 \ddot{H}H - 17 \dot{H}^2 + 38 \dot{H}H^2 + 19 H^4) z_0^4 \nonumber \\
	&& + \frac{1}{240} (3 \tdot{H} + 26 \ddot{H}H + 17 \dot{H}^2 + 52 \dot{H}H^2 + 1 H^4) a^2 \vec{z}^2 z_0^2 \nonumber \\
	&& + \frac{1}{480} (4 \ddot{H}H + 3 \dot{H}^2 + 36 \dot{H}H^2 + 29 H^4) a^4 (\vec{z}^2)^2 + \mathcal{O}(z^5)
\eeqan
which can be checked to be symmetric under interchange of $x$ and $y$ by a straight forward calculation. This equation has been derived using computer algebra support, but has been validated up to second order by hand calculations.

From this we obtain

\begin{Satz}
\label{3.1theo}
The most singular order of the Hadamard distribution is given by
\begin{align}
- 4 \pi^2 \mathfrak{H'} (z) |_{z^0 = 0} =& - \frac{1}{a^2 \vec{z}^2} - \frac{1}{12} \left(\dot{H} + 2 H^2 \right) \nonumber\\
 & - \frac{1}{1440}(12\ddot{H}H + 9 \dot{H}^2 + 86 \dot{H}H^3 + 51 H^4) a^2 \vec{z}^2 + \mathcal{O}((\vec{z}^2)^2)\\
- 4 \pi^2 \dot{\mathfrak{H'}} (z) |_{z^0 = 0} =&\frac{H}{a^2 \vec{z}^2} - \frac{1}{24} (\ddot{H} + 4 \dot{H} H) + \mathcal{O}(\vec{z}^2)\\
- 4 \pi^2 \pabl{}{x^0} \pabl{}{y^0} \mathfrak{H'}(z) |_{z^0 = 0} =& \frac{2}{a^4 (\vec{z}^2)^2} - \frac{H^2}{a^2 \vec{z}^2} \nonumber\\
 & - \frac{1}{240} (4 \tdot{H} + 16 \ddot{H}H + 27 \dot{H}^2 + 30 \dot{H}H^2 + 17 H^4) + \mathcal{O}(\vec{z}^2)
\end{align}
\end{Satz}

The proof is done by straightforward calculations where the occurring fractions in powers of $Z = \vec{z}^2$ are expanded in a power series like
\beqan
\frac{L + M Z + N Z^2}{P + Q Z + S Z^2} \is \frac{L}{P} + \frac{M P - L Q}{P^2} Z + \frac{N P^2 - L P S - M P Q + L Q^2}{P^3} Z^2
\eeqan

Even if placeholders are inserted for the coefficients of the powers of $Z$ and $z_0$ in $\sigma$ and $u$ the calculations remain very lengthy, as the coefficient's derivatives have to be taken into account. Therefore we will not show them here explicitly.

We calculate the sub leading terms of the Hadamard parametrix in position space to the order relevant to calculate the homogeneous term of the energy momentum tensor. This means we use the Hadamard recursion
\beqan
2 g^{\mu\nu}(x) (\partial^x_{\mu} \sigma) (\partial^x_{\nu} v_0) + (\square^x \sigma - 2) v_0 \is -(\square^x + m^2 - \xi R(x))u\\
2 g^{\mu\nu}(x) (\partial^x_{\mu} \sigma) (\partial^x_{\nu} v_1) + (\square^x \sigma) v_1 \is -(\square^x + m^2 - \xi R(x))v_0
\eeqan
to calculate
\beqan
v_0 \is -\frac{1}{2}\left( \left( \frac{1}{6} - \xi \right)R + m^2 \right) + \frac{1}{4} (\ddot{H} +4 \dot{H} H ) z_0 \nonumber\\
&&+\frac{1}{240} \Big((21-120\xi) \tdot{H} + (87-480\xi) \ddot{H}H + (54-300\xi) \dot{H}^2 \nonumber\\
&&\qquad - (76-540\xi) \dot{H}H^2 - (58-360\xi) H^4 + 30 m^2(\dot{H} + H^2)\Big) z_0^2 \nonumber\\
&&+\frac{1}{240} \Big(- \tdot{H} + (3-60\xi) \ddot{H}H + (6-60\xi) \dot{H}^2 + (76-540\xi) \dot{H}H^2 \nonumber\\
&&\qquad + (58-360\xi) H^4 - 10 m^2(\dot{H} + 3 H^2)\Big) a^2 \vec{z}^2 + \mathcal{O}(z^3) \label{v0result}
\eeqan
and
\begin{align}
v_1 =& \frac{1}{120} \Big( (1-5\xi)\square R + \frac{5}{12} (1-6\xi)^2 R^2 - 2 (\dot{H} + H^2)H^2 + 5 (1-6\xi) R m^2 \nonumber\\
&+ 15 m^4 \Big) + \mathcal{O}(z) \label{v1result}
\end{align}

\section{Mode Expansion}

Two things remain to be done: The singular terms in (\ref{parametrix}) have to be related to the mode expansion $G_{\phi\phi,k}$ etc. of the two point functions and the homogeneous term of the energy momentum tensor has to be calculated. 

First, the expressions in Theorem \ref{3.1theo} can easily be Fourier transformed with a method explained e.g. in \cite{Con}. Using the distributional Fourier transforms \linebreak $\mathscr{F}(|\vec{z}\, |^{-2})(\vec k) = 2 \pi^2 k^{-1}$, $\mathscr{F}( |\vec{z}\, |^{-4})(\vec k)  = - \pi^2 k$
with $k = |\vec{k}|$ and ignoring for now homogeneous terms that will only contribute terms proportional to the delta distribution in the zero mode and positive order terms in $Z$, we get the mode expressions
\beqan
\mathscr{F} \left(\mathfrak{H'}|_{z^0 = 0} \right)(\vec{k}) \is a^{-2} \frac{1}{2k}\\
\mathscr{F} \left(\dot{\mathfrak{H'}}|_{z^0 = 0} \right)(\vec{k}) \is - a^{-2} \frac{H}{2k}\\
\mathscr{F} \left(\pabl{}{x^0} \pabl{}{y^0} \mathfrak{H'}|_{z^0 = 0} \right)(\vec{k}) \is a^{-4} \frac{k}{2} + a^{-2} \frac{H^2}{2k} \label{fourier}
\eeqan

Here it needs to be taken into account that all singular orders in $\vec z^2$ that are encountered  are locally integrable functions on $\mathbb{R}^3$ with the exception of the $\frac{1}{z^4}$ term in $\pabl{}{x^0} \pabl{}{y^0} \mathfrak{H'}|_{z^0 = 0}$. Continuation of the $ \pabl{}{x^0} \pabl{}{y^0}[ \left<\phi(x)\phi(y)\right>-\mathfrak{H}_n(x,y)]|_{z^0 = 0}$ for $n\geq 2$ to $\vec z=0$ thus implicitly induces a rotational invariant regularization of $\frac{1}{\vec z^4}$, as it can be written as a difference of already regularized terms and the newly found continuous function. The analytic regularization \cite{Con} of $|\vec z|^\zeta$ is well defined except for poles at $\zeta=-3,-5,-7,\ldots$. It is clearly rotational invariant. Furthermore $|\vec z|^\zeta$ with $\zeta=-4$ is the only invariant extension of $\frac{1}{\vec z^4}$ on $\mathbb{R}^3$ that preserves the scaling degree, as there are no rotation invariant linear combinations of first order derivative of the delta distribution at zero. But the difference of two extensions with the given properties has to be such a linear combination, \cite{Con}. Hence we conclude that it has to coincide with the regularization induced by the continuation of $\pabl{}{x^0} \pabl{}{y^0}[ \left<\phi(x)\phi(y)\right>-\mathfrak{H}_n(x,y)]|_{z^0 = 0}$ to $\vec z=0$.  The Fourier transform (\ref{fourier}) can thus be calculated based on the analytic regularization prescription \cite{Con}.  

We now turn to the mode expansion of the sub leading terms. Expanding the logarithmic factor with the help of (\ref{worldFunctionExpansion}) yields
\beq
\left.\log\left(-\frac{\sigma(x,y)}{\lambda^2}\right)\right|_{z^0=0}=\log\left(\frac{a^2}{\lambda^2}\right)+\log(\vec z^2)-\sum_{n=1}^\infty \left(-\frac{1}{12}H^2a^2\vec z^2+{\cal O}(\vec z ^4)\right)^n,
\eeq  
where the infinite sum can be truncated after a few relevant terms. Inserting (\ref{worldFunctionExpansion}) to expand the powers $\sigma(x,y)^k$ in (\ref{parametrix}), we easily see that the singular contribution from the second term in (\ref{parametrix}) to $[ \left<\phi(x)\phi(y)\right>-\mathfrak{H}_n(x,y)]|_{z^0 = 0}$ is 
\beq
\label{logTermExpansion}
\frac{1}{4\pi^2} v_0((t,\vec x),(t,\vec y))\left(\log\left(\frac{a^2}{\lambda^2}\right)+\log(\vec z^2)\right).    
\eeq
Likewise, using $\frac{1}{\sigma(x,y)}|_{z^0=0}=-\frac{1}{\vec z^2}\sum_{n=0}^\infty (\frac{1}{12} H^2a^2\vec z^2+{\cal{O}}(\vec z^4))^n$, it is easily shown that the expansion of derivatives of the second term in (\ref{parametrix}) with respect to $\frac{\partial}{\partial x^0}$ and $\frac{\partial^2}{\partial x^0\partial y^0}$ at equal time into singular orders in $\vec z^2$ which is truncated after some sufficiently high order contains only terms $\sim \log(\vec z^2)$ and $\sim (\vec z^2)^k, k\geq 0$, where the coefficients are made of $\log(\lambda)$, $\log(a)$ and $U_k((t,\vec x),(t,\vec y))$-terms and their derivatives w.r.t. $\vec z^2$ along with their time derivatives; for the case with the two time derivatives an additional term with the singularity structure $\frac{1}{\vec z^2}$ needs to be taken into account as well. 

One could calculate the coefficients of the singular terms described above by calculating the derivatives of (\ref{v0result}) and (\ref{v1result}) and Fourier transforming the result by hand. However these calculations are very tedious, therefore we proceed differently here. As we have analyzed the general shape of the terms to account for, we can Fourier transform these terms and include them into the mode expansion with up to now undetermined coefficients. Doing so, we use that $\mathscr{F}(\log(|z|))(\vec k)= -4\pi^{3/2} \Gamma(\frac{3}{2}) k^{-3}$. Note that this distribution at $0$ does not extend a locally integrable function and hence requires a regularization prescription. The details can be found in appendix \ref{Aapp}. The remaining singular orders that occur have mode expansion $\mathscr{F}(\vec z^2\log(|z|))(\vec k)= 24\pi^{3/2} \Gamma(\frac{3}{2}) k^{-5}$ and $\mathscr{F}((\vec z^2)^n)(\vec k)=(2\pi)^3(-\Delta)^n\delta(\vec k) $, where we again neglect the latter for the time being.         

Next we multiply the leading terms in the singular order expansion with the powers of $a$ to get expressions similar to the two-point function
\beqan
\mathfrak{H'}_{\phi\phi} \is \frac{a^{-2}}{2k}\\
\mathfrak{H'}_{(\phi\pi)} \is - \frac{a H}{2k}\\
\mathfrak{H'}_{\pi\pi} \is \frac{a^2 k}{2} + \frac{a^4 H^2}{2k}
\eeqan
So, taking now into account the discussion above, we consider the ansatz
\beqan
\mathfrak{H}_{\phi\phi} \is \frac{a^{-2}}{2k} + \frac{\alpha_3}{2k^3} + \frac{\alpha_5}{2k^5} + \mathcal{O}\left( k^{-7} \right)\\
\mathfrak{H}_{(\phi\pi)} \is - \frac{a H}{2k} + \frac{\beta_3}{2k^3} + \frac{\beta_5}{2k^5} + \mathcal{O}\left( k^{-7} \right)\\
\mathfrak{H}_{\pi\pi} \is \frac{a^2 k}{2} + \frac{a^4 H^2}{2k} + \frac{\gamma_1}{2k} + \frac{\gamma_3}{2k^3} + \mathcal{O}\left( k^{-5} \right)
\eeqan
that captures all possible singular orders except for those concentrated in the zero mode.

Note that $\mathfrak{H}_n$ fulfills the Klein-Gordon equation up to terms that vanish in the coincidence limit and a zero mode term. Plugging our ansatz of the singular order expansion into the system of equations (\ref{modeMotion}) we obtain a couple of equations that help us to determine the unspecified coefficients

\beqan
\frac{\dot{\alpha}_3}{2k^3} + \frac{\dot{\alpha}_5}{2k^5} \is \frac{2 a^{-3}\beta_3}{2k^3} + \frac{2 a^{-3}\beta_5}{2k^5}\\
 - \frac{a (\dot{H} + 2 H^2)}{2k} + \frac{\dot{\beta}_3}{2k^3} \is - \frac{(-\xi R + m^2)a + a \alpha_3}{2k} - \frac{(-\xi R + m^2) a^3 \alpha_3 + a \alpha_5}{2k^3}\nonumber\\
	&& + \frac{a^{-3}\gamma_1}{2k} + \frac{a^{-3}\gamma_3}{2k^3}\\
\frac{2 a^4 H (\dot{H} + 2 H^2)}{2k} + \frac{\dot{\gamma}_1}{2k} + \frac{\dot{\gamma}_3}{2k^3} \is - \frac{- 2(-\xi R + m^2)a^4H + 2 a \beta_3}{2k}\nonumber\\&& - \frac{2(-\xi R + m^2)a^3 \beta_3 + 2 a \beta_5}{2k^3}
\eeqan
where we suppress undetermined orders.

Assuming these equations to hold order by order in $k$ we get two sets of equations
\beqan
\dot{\alpha}_3 \is 2 a^{-3} \beta_3 \label{a1}\\
\dot{\gamma}_1 \is -2 a \beta_3 + 2 a^4 H \left( \left( \frac{1}{6} - \xi \right) R + m^2 \right) \label{a2}\\
\alpha_3 \is a^{-4} \gamma_1 - \left(\frac{1}{6} - \xi \right) R - m^2 \label{a3}
\eeqan
and analogously
\beqan
\dot{\alpha}_5 \is 2 a^{-3} \beta_5\\
\dot{\gamma}_3 \is - 2 a \beta_5 - 2 a^3 \beta_3 \left( -\xi R + m^2 \right)\\
\alpha_5 \is a^{-4} \gamma_1 - a^{-1} \dot{\beta}_3 - a^2 \alpha_3 \left( -\xi R + m^2 \right) \label{b2}
\eeqan

We now want to solve these differential equations. The equation (\ref{a1})+(\ref{a2})-$\partial_t$(\ref{a3}) leads to
\beq
\dot{\gamma}_1 - 2 H \gamma = \frac{a^4}{2} \left( \frac{1}{6} - \xi \right) \left( \dot{R} + 2 H R \right) + a^4 H m^2
\eeq
which can be solved using standard methods. The solutions for $\alpha_3$ and $\beta_3$ can then be calculated straightforwardly without solving differential equations.

\begin{Lemma}
\label{4.1lem}
The solution to the system (\ref{a1})-(\ref{a3}) is
\beqan
\alpha_3 \is - \frac{1}{2} \left( \left(\frac{1}{6} - \xi \right) R + m^2 \right) + A a^{-2}\\
\beta_3 \is - \frac{a^3}{4} \left(\frac{1}{6} - \xi \right) \dot{R} - A H a\\
\gamma_1 \is \frac{a^4}{2} \left(\left(\frac{1}{6} - \xi \right) R + m^2 \right) + A a^2
\eeqan
where $A$ is some constant.
\end{Lemma}

\begin{proof}
A differential equation of the form
\beq
\dot{x}(t) - 2 f(t) x(t) = g(t)
\eeq
has the solution
\beq
x(t) = \int_0^t \limits g(t') \exp{\left( 2 \int_{t'}^t \limits f(t'') dt'' \right)} dt' + A \exp{\left( 2 \int_0^t \limits f(t') dt' \right)}
\eeq
with $A$ being constant. In our case $f(t) = H$, such that $\exp{\left( 2 \int_{t_1}^{t_2} \limits H(t') dt' \right)} = \frac{a^2(t_1)}{a^2(t_2)}$.

We get the solution
\beqan
\gamma_1 \is a^2 \int_0^t \limits \left( \frac{1}{2} \left( \frac{1}{6} - \xi \right) \left( \dot{R} + 2 H R \right) + H m^2 \right) a^2 dt' + A a^2\\
\is \frac{a^2}{2} \int_0^t \limits \partial_{t'} \left[ \left( \left( \frac{1}{6} - \xi \right) R + m^2 \right) a^2 \right] dt' + A a^2\\
\is \frac{a^4}{2} \left(\left(\frac{1}{6} - \xi \right) R + m^2 \right) + A a^2
\eeqan
where the constant $A$ absorbs all constant terms and thus changes from line to line.

\beqan
2 a \beta_3 \is - \dot{\gamma}_1 + 2 a^4 H \left( \left( \frac{1}{6} - \xi \right) R + m^2 \right)\\
\is - 2 a^4 H \left(\left(\frac{1}{6} - \xi \right) R + m^2 \right) - \frac{a^4}{2} \left(\frac{1}{6} - \xi \right) \dot{R} - 2A a^2 H\\
&& + 2 a^4 H \left( \left( \frac{1}{6} - \xi \right) R + m^2 \right)\\
\is - \frac{a^4}{2} \left(\frac{1}{6} - \xi \right) \dot{R} - 2 A a^2 H
\eeqan
and
\beqan
\alpha_3 \is a^{-4} \gamma_1 - \left(\frac{1}{6} - \xi \right) R - m^2\\
\is \frac{1}{2} \left(\left(\frac{1}{6} - \xi \right) R + m^2 \right) + A a^{-2} - \left(\frac{1}{6} - \xi \right) R - m^2\\
\is - \frac{1}{2} \left(\left(\frac{1}{6} - \xi \right) R + m^2 \right) + A a^{-2}
\eeqan
yield the other two coefficients without solving differential equations.
\end{proof}

Applying the same strategy as before we get the solutions for the second system of equations
\beqan
\alpha_5 \is \frac{a^2}{8} \Bigg\{ \left(\frac{1}{6} - \xi \right) \left(\ddot{R} + 5 H \dot{R} - \dot{H} R \right) - 3 \left(\frac{1}{6} - \xi \right) \xi R^2 \\
&&- \left(4 \dot{H} + 6 H^2 + 6 \xi R \right) m^2 - 3 m^4 \Bigg\}\\
\beta_5 \is \frac{a^3}{2} \dot{\alpha}_5\\
\gamma_3 \is - \frac{a^6}{8} \Bigg\{ \left(\frac{1}{6} - \xi \right) \left(\ddot{R} + H \dot{R} + \dot{H} R \right) - \left(\frac{1}{6} - \xi \right) \xi R^2\\
&& - 2 \left( H^2 + \xi R \right) m^2 - m^4 \Bigg\}
\eeqan
where we suppress the integration constant as it vanishes for a pure state due an argument given below.

\begin{Satz}
\label{positivityThm}
The singular part of the Hadamard modes is equal to the singularities of the two-point function of a pure state if the integration constant $A$ in Lemma \ref{4.1lem} vanishes. That is,
\begin{itemize}
\item (i) $\forall k : \; \mathfrak{H}_{\phi\phi,k} \mathfrak{H}_{\pi\pi,k} - \mathfrak{H}_{(\phi\pi),k}^2 = \frac{1}{4} + \mathcal{O}\left( k^{-6} \right)$
\item (ii) $\forall k : \; \mathfrak{H}_{\phi\phi,k} + \mathfrak{H}_{\pi\pi,k} > 0 + \mathcal{O}\left( k^{-3} \right)$
\end{itemize}
holds if and only if $A=0$.
\end{Satz}

\begin{proof}
The first claim is equivalent to
\beqan
a^2 H^2 + a^{-2} \gamma_1 + a^2 \alpha_3 - a^2 H^2 \is 0 \label{null1}\\
a^{-2} \gamma_3 + a^2 \alpha_5 + a^4 H^2 \alpha_3 + \alpha_3 \gamma_1 + 2 aH \beta_3 \is 0 \label{null2}
\eeqan
where equation (\ref{null1}) consists of the terms of order $k^{-2}$ and equation (\ref{null2}) represents the order $k^{-4}$. Equation (\ref{null1}) gives $2A = 0$. Analogously equation (\ref{null2}) forces the other integration constant to vanish.

To show that the second claim is satisfied we use the fact that $\mathfrak{H}$ is a solution of the Klein-Gordon equation (\ref{modeMotion}) up to $\mathcal{O}\left( k^{-5} \right)$. Then we use a form of the well known deformation argument of Wald, which is explained in \cite{WalB}. We deform the spacetime for early times such that $a$ is constant at early times. Thus the claim holds exactly at early times, as we just get the Minkowski two-point function. Then we propagate $\mathfrak{H}$ with Equation (\ref{modeMotion}) to some later time, where the spacetime is undeformed. Thus the claim is still satisfied at that time because the first claim shown above ensures that $\mathfrak{H}_{\phi\phi,k} \mathfrak{H}_{\pi\pi,k} > 0 + \mathcal{O}\left( k^{-6} \right)$. Finally, we undeform the spacetime at early times and can thus show the claim to hold for all times using Equation (\ref{modeMotion}).
\end{proof}

From Theorem \ref{positivityThm} we obtain a (tentative) parameterization of initial conditions of the field degrees of freedom for the equation of motion for pure states. For $k>1$ we set $G_{\phi\phi,k}=\mathfrak{H}_{\phi\phi,k}+a(k)$, $G_{\phi\pi,k}=\mathfrak{H}_{(\phi\pi),k}+b(k)$ and $G_{\pi\pi,k}=(\frac{1}{4}+G_{(\phi\pi),k}^2)/G_{\phi\phi,k}$ where $a(k)k^{7}$ and $b(k)k^{7}$ are $C^{\infty}_b((1,\infty))$. Furthermore, $a(k)$ needs to be choosen such that $G_{\phi\phi,k}>0$ for all such $k$. This can always be achieved as the leading term $a^{-2}/k$ is positive. For $k\leq 1$, these functions need to be continued to $C_b^\infty([0,\infty))$ such that the relations $G_{\phi\phi,k}+G_{\pi\pi,k}>0$ and  (\ref{sCCR}) are preserved. Obviously with such initial conditions the trace of the initial energy momentum tensor is finite. In particular, Hadamard states fall into this class.

From the construction of $\mathfrak{H}_{\phi\phi,k}$, $\mathfrak{H}_{(\phi\pi),k}$ and $\mathfrak{H}_{\pi\pi,k}$ it follows that these quantities fulfill  (\ref{modeMotion}) up to order $k^{-7}$ and $k^{-5}$, respectively. Thus, if the differece of $G_{\phi\phi,k}$ and $\mathfrak{H}_{\phi\phi,k}$ (multiplied by $k^2$) is integrable for large $k$, this property will prevail after an infinitisimal time step. The argument for the difference of $G_{\pi\pi,k}$ and $\mathfrak{H}_{\pi\pi,k}$ is analoguous.   

For non pure states, a positive $c(k)$ can be added to $G_{\pi\pi,k}$ such that $k^5c(k)\in C^b([0,\infty)]$.

\section{Equation of motion}
We begin with the easier case where the field is conformally coupled to the mean curvature, i.e. $\xi=1/6$. In this case, (\ref{motion2}) simplifies considerably, as it only contains the field dependent term $\langle \phi\phi\rangle_{\omega,\lambda,\xi}$. This term, up to zero mode contributions, is 
\beqan
\frac{1}{8\pi^3}\int_{\mathbb{R}^3} \left[G_{\phi\phi,k}-\frac{a^{-2}}{2k}+\frac{m ^2}{4k^3}\right] dk,
\eeqan 
since $\alpha_3=-\frac{m^2}{2}$. The renormalization prescription of $\int_{\mathbb{R}^3} k^{-3} dk$ at $0$ is given in Appendix \ref{Aapp}.

The zero mode terms in the conformally coupled case are 
\beq
\frac{1}{72}\,R-\frac{1}{4\pi^2}v_0\log\left(\frac{a^2}{\lambda^2}\right),
\eeq
where the first term stems from the first equation in Theorem \ref{3.1theo} and the second one from (\ref{logTermExpansion}). In the conformally coupled case, $v_0(x,x)$ takes the simple form 
\beq
v_0=-\frac{m^2}{2}.
\eeq
Wrapping up, we obtain 

\begin{Satz}
\label{5.1theo}
For the case of conformal coupling, the equations of motion (\ref{motion2}) is
\beqan
\label{motion3}
-R&=&8\pi G\left\{ \frac{m^2}{8\pi^3}\int_{\mathbb{R}^3} \left[G_{\phi\phi,k}-\frac{a^{-2}}{2k}+\frac{m ^2}{4k^3}\right]dk+m^2\left(\frac{1}{72}+c'\right)R  \right.\nonumber\\
&&+\frac{1}{240\,\pi^2}\left(\dot HH^2+H^4\right) +\left(\frac{1}{2880\,\pi^2}+c''\right)\Box R+\frac{m^4}{4\pi^2} \log\left(\frac{a^2}{\lambda^2}\right)\nonumber\\
&&\left.-m^4\left(\frac{1}{32 \,\pi^2}-c\right)\right\}
\eeqan
where the $k^{-3}$ integral is regularized at zero as explained in Appendix \ref{Aapp}.
\end{Satz}

Wald's fifth axiom \cite{Wal77,Wal78} states that the equation of motion should be a second order equation. (\ref{motion3}) is second order if and only if the renormalization constants $A,B,C$ and $D$ in (\ref{reno}) are chosen such that $c''=-1/(2880\pi^2)$. Note however, that for this value of $c''$ the joint system differential equations (\ref{modeMotion}) and (\ref{motion3}) is an implicit (infinite dimensional) system of differential equations and for $c''\not=-1/(2880\pi^2)$ it is explicit, which is of advantage for proving existence and numerical solutions. 

The $\log(a)$ term that seems to be missing in \cite{Pina} is due to the fact that the infra red prescription we employ for $\int_{\mathbb{R}^3}\frac{f(k)}{k^3}dk$ (cf. Appendix A) is not $a$ dependent, whereas the prescription used in \cite{Pina} is, leading to the absence of a stand alone $\log(a)$ term in the latter case.

For $\xi\not=1/6$ the homogeneous terms of the Hadamard parametrix take a much more complicated form, as second derivatives of the parametrix must be taken into account. Ignoring a prefactor $\frac{8\pi G}{4\pi^2}$ we have to calculate zero mode terms of 
\begin{align}
\mathfrak{C}_{\xi,m}=&(6\xi-1) \left(\partial_{x_0} \partial_{y_0} + a^{-2} \triangle \right) \left(- 2 \frac{u}{2\sigma} - v_0 \log{(2\sigma)} - \frac{1}{2} v_1 2\sigma \log{(2\sigma)} \right) \nonumber\\
&- \frac{1}{12} \left[(2-6\xi)m^2 + (6\xi-1)\xi R \right] (\dot{H} + 2 H^2) \nonumber\\
&+ \log{\left(\frac{a^2}{\lambda^2}\right)} \Bigg( (6\xi-1) \left(\partial_{x_0} \partial_{y_0} + a^{-2} \triangle \right) \left( - v_0 - \frac{1}{2} v_1 2\sigma \right) \nonumber\\
&- \left[(2-6\xi)m^2 + (6\xi-1)\xi R \right] v_0 \Bigg)
\end{align}
in the coincidence limit $x\to x'$. This can be done using the expansions of $u$, $v_0$, $v_1$ and $\sigma^2$ given in Section 3, applying the differential operator, setting $z^0=0$ and extracting terms of $\vec z$-order zero. With a computer aided calculation this yields the result
\begin{align}
\mathfrak{C}_{\xi,m}=& - \frac{1}{30} \left( 4 \tdot{H} + 53 \ddot{H}H + 11 \dot {H}^2 + 141 \dot{H}H^2 + 3 H^4 \right) \nonumber\\
&+\frac{1}{5} \xi \left( 9 \tdot{H} + 113 \ddot{H}H - 44 \dot{H}^2 + 11 \dot{H}H^2 - 277 H^4 \right) \nonumber\\
&- 6 \xi^2 \left( \tdot{H} + 12 \ddot{H}H - 14 \dot{H}^2 - 38 \dot{H}H^2 - 68 H^4 \right) - 108 \xi^3 \left(\dot{H} + 2 H^2 \right)^2 \nonumber\\
&+ \frac{m^2}{6} \left( \left( 21 \dot{H} + 44 H^2 \right) - 6\xi \left( 3 \dot{H} + 8 H^2 \right) - 36 \xi^2 \left( \dot{H} + 2 H^2 \right) \right) \nonumber\\
&+ \frac{m^4}{2} \left( 1 - 6 \xi \right) \nonumber\\
&+ \log{\left(\frac{a^2}{\lambda^2}\right)} \Bigg( - \frac{1}{60} \left( 3 \tdot{H} + 21 \ddot{H}H + 42 \dot {H}^2 + 152 \dot{H}H^2 + 116 H^4 \right) \nonumber\\
&+ \frac{1}{5} \xi \left( 4 \tdot{H} + 28 \ddot{H}H +31 \dot{H}^2 + 106 \dot{H}H^2 + 58 H^4 \right) \nonumber\\
&- \xi^2 \left( 3 \tdot{H} + 21 \ddot{H}H - 6 \dot{H}^2 - 36 \dot{H}H^2 - 72 H^4 \right) - 108 \xi^3 \left(\dot{H} + 2 H^2 \right)^2 \nonumber\\
&+ \frac{3}{2} m^2 \left( 1 - 6 \xi \right) \left( \dot{H} + 2 H^2 \right) - m^4 \left( 1 - 3 \xi \right) \Bigg)
\end{align}

Although this result looks very tedious, one can simplify the result subsuming terms proportional to $\Box R$, and $m^2R$ into the renormalization degrees of freedom and the $H$ independent term into a redefinition of the scale parameter $\lambda$.

The simplified form is
\begin{align}
\mathfrak{C}_{\xi,m}=& - \frac{1}{30} \left( -\frac{2}{3} \square R - \frac{25}{6} \dot{R} - \frac{5}{36} R^2 + 13 \dot{H}H^2 + 23 H^4 \right) \nonumber\\
&+\frac{1}{5} \xi \left( -\frac{3}{2} \square R - \frac{25}{3} \dot{R} - \frac{20}{9} R^2 + 23 \dot{H}H^2 + 43 H^4 \right) \nonumber\\
&- 6 \xi^2 \left( -\frac{1}{6} \square R - \frac{5}{6} \dot{R} - \frac{1}{2} R^2 + 2 \dot{H}H^2 + 4 H^4 \right) - 3 \xi^3 R^2 \nonumber\\
&+ \frac{m^2}{6} \left( \left( - \frac{7}{2} R + 2 H^2 \right) - 6\xi \left( - \frac{1}{2} R + 2 H^2 \right) + 6 \xi^2 R \right) + \frac{m^4}{2} \left( 1 - 6 \xi \right) \nonumber\\
&+ \log{\left(\frac{a^2}{\lambda^2}\right)} \Bigg( - \frac{1}{60} \left( -\frac{1}{2} \square R + \frac{5}{6} R^2 - 4 \dot{H}H^2 - 4 H^4 \right) \nonumber\\
&+ \frac{1}{5} \xi \left( -\frac{2}{3} \square R + \frac{5}{12} R^2 - 2 \dot{H}H^2 - 2 H^4 \right) \nonumber\\
&- \xi^2 \left( -\frac{1}{2} \square R - \frac{1}{2} R^2 \right) - 3 \xi^3 R^2 + \frac{1}{4} m^2 \left( 1 - 6 \xi \right) R - m^4 \left( 1 - 3 \xi \right) \Bigg)
\end{align}

Plugging in the calculated homogeneous term we obtain the equation of motion for the general case.

\begin{Satz}
The equation of motion in the general case is given by the following expression
\beqan
\label{motion4}
-R&=&8\pi G\left\{\left(6\xi-1\right)\left(\frac{a^{-6}}{8\pi^3}\int_{\mathbb{R}^3} \left[G_{\pi\pi,k}-\frac{a^2k}{2}-\frac{a^4H^2+\gamma_1}{2k}-\frac{\gamma_3}{2k^3}\right] dk\right.\right.\nonumber\\
&&\left.+\frac{a^{-2}}{8\pi^3}\int_{\mathbb{R}^3}\left[k^2 G_{\phi\phi,k}-
\frac{a^{-2}k}{2}-\frac{\alpha_3}{2k}-\frac{\alpha_5}{2k^3}\right]dk\right)\nonumber \\
&&+\frac{\xi}{8\pi^3} \left[6m^2+(6\xi-1)R\right]\int_{\mathbb{R}^3}\left[G_{\phi\phi,k}-\frac{a^{-2}}{2k}-\frac{\alpha_3}{2k^3}\right]dk\nonumber \\
&&\left. - \frac{1}{4\pi^2} \mathfrak{C}_{\xi,m} + \frac{36\xi-5}{4\pi^2}v_1 +cm^4+c'm^2R+c''\Box R\right\}
\eeqan
The regularization of the integrals over $k ^{-3}$ at zero is given in Appendix \ref{Aapp}.   
\end{Satz}
We would like to point out that due to the 4th time derivatives in $\gamma_3$ and $\alpha_5$, (\ref{motion4}) again leads to an implicit system of differential equations, as the leading order time derivative can not be isolated, regardless of the renormalization degrees of freedom and the specific form of $\mathfrak{C}_{\xi,m}(t)$. Also, Wald's fifth axiom can never be fulfilled in this case.
\appendix

\section{Regularization prescription on $k^{-3}$}
\label{Aapp}
Here we give the details of the regularization of the seemingly infra red divergent integrals 
\begin{equation}
\int_{\mathbb{R}^3} \left[G_{\phi\phi,k}-\frac{a^{-2}}{2k}+\frac{m ^2}{4k^3}\right] dk
\end{equation}
and related. For any Schwartz test function $\chi$ with $\chi(0)=1$ This integral can be re-written as
\begin{equation}
\int_{\mathbb{R}^3} \left[G_{\phi\phi,k}-\frac{a^{-2}}{2k}+\frac{m ^2}{4k^3}\right](1-\chi(k)) + \left[G_{\phi\phi,k}-\frac{a^{-2}}{2k}\right]\chi(k)\, dk+\int_{\mathbb{R}^3} \frac{m ^2}{4k^3}\chi(k) dk
\end{equation} 
Here the first integral is regular. The second -- seemingly infra red divergent -- integral needs to be properly traced back to its definition as distributional Fourier transform of $\log(r)$ smeared with the test function $\chi$, which gives a regular integral after shifting the Fourier transform to the test function.  

Let us quickly calculate the Fourier transform in the sense of tempered distributions of the locally integrable function $\log(r)$, $r=|\vec x|$ for the convenience of the reader. Let $\check \varphi(\vec x)=\varphi(-\vec x)$, then
\beqan
\label{fourierLog}
\left<\mathscr{F}(\log(r))(k),\varphi\right>&=& \left< \log(r),\mathscr{F}(\check \varphi)\right>=\lim_{\zeta\nearrow 0}\frac{d}{d\zeta} \left<r^\zeta,\mathscr{F}(\check \varphi)\right>\nonumber\\
&=&\lim_{\zeta\nearrow 0}\frac{d}{d\zeta}\left[2^{\zeta+3}\pi^{3/2} \frac{\Gamma\left(\frac{\zeta+3}{2}\right)}{\Gamma\left(-\frac{\zeta}{2}\right)}\left<k^{-\zeta-3},\varphi\right>\right]
\eeqan 
Since $\zeta$ approaches $0$ from below, $k^{-\zeta-3}$ is locally integrable. Furthermore, $\left<k^{-\zeta-3},\varphi\right>$ can be analytically continued in $\zeta$ to $\mathbb{C}\setminus\mathbb{N}_0$ \cite{Con}. The Laurent series of $k^{-\zeta-3}$ in $\zeta=0$ has a pole proportional to $\delta_0$. Hence, for $\varphi(0)=0$, $\left<k^{-\zeta-3},\varphi\right>$ is analytic at $\zeta=0$ and in particular has a finite $\zeta$ derivative. So let us from now on assume $\varphi(0)=0$ and perform the $\zeta$ derivative at the right hand side of (\ref{fourierLog}). Due to the pole of $\Gamma\left(-\frac{\zeta}{2}\right)$ at $\zeta=0$, the only term that contributes in the limit $\zeta\nearrow 0$ is
\beq
\lim_{\zeta\nearrow 0} 2^{\zeta+3-1}\pi^{3/2} \frac{\Gamma\left(\frac{\zeta+3}{2}\right)\Gamma'\left(-\frac{\zeta}{2}\right)}{\Gamma\left(-\frac{\zeta}{2}\right)^2}\left<k^{-\zeta-3},\varphi\right>=-4\pi^{3/2}\Gamma\left(\frac{3}{2}\right) \left< k^{-3},\varphi\right>,
\eeq  
where we used the Laurent series expansion of $\Gamma$ at zero. Since, by dominated convergence, $\left< k^{-3},\varphi\right>=\int_{\mathbb{R}^3}k^{-3}\varphi(\vec x) d\vec x$, we see that $\mathscr{F}(\log(r))(k)$ is a regularization of $-4\pi^{3/2}\Gamma\left(\frac{3}{2}\right)k^{-3}$. Let now $\varphi(0)\not=0$, and $\chi(k)=e^{-k^2}$. Since $\varphi(\vec x)=[\varphi(\vec x)-\varphi(0)e^{-k^2/2}]+\varphi(0)e^{-k^2/2}=\varphi_1(\vec x)+\varphi(0)e^{-k^2/2}$. We have already shown how to evaluate $\mathscr{F}(\log(r))(k)$ on $\varphi_1$. It remains to calculate
\beq
\left<\mathscr{F}(\log(r)),\varphi(0)e^{-k^2/2}\right>=\frac{\varphi(0)}{(2\pi)^{3/2}}\left<\log(r),e^{-k^2/2}\right>=\frac{\varphi(0)}{(2\pi)^{1/2}}\Gamma'\left(\frac{3}{2}\right).
\eeq

\vspace{1cm}

\noindent {\sc Benjamin Eltzner}\\
Max-Planck Institute for Mathematics in the Sciences\\
Inselstr. 22, D-04103 Leipzig\\
eltzner@mis.mpg.de\\

\vspace{.5cm}

\noindent {\sc Hanno Gottschalk}\\
Bergische Universit\"at Wuppertal, \\
Fachgruppe Mathematik, Gau\ss stra\ss e 20\\
D-42119 Wuppertal \\
gottschalk@math.uni-wuppertal.de

\end{document}